\documentclass[11pt]{article}
\usepackage{amsmath}
\usepackage{subfigure}
\usepackage{amsfonts}
\usepackage{float}
\usepackage{amsthm}
\usepackage{amssymb, setspace}
\usepackage{color,graphicx, epsfig, geometry, hyperref,fancyhdr}
\usepackage{mathrsfs}
\usepackage[T1]{fontenc}
\usepackage{latexsym,amssymb,amsmath,amsfonts,amsthm}
\usepackage{graphicx}
\usepackage{color}
\usepackage{soul}
\usepackage{caption}
\usepackage{dsfont}
\usepackage{multirow}
\usepackage{rotating}
\newtheorem{theorem}{Theorem}[section]

\newtheorem{lemma}{Lemma}[section]

\newcommand{\R}{\mathbb{R}}
\newcommand{\C}{\mathbb{C}}

\newcommand{\grad}{\nabla}

\newcommand{\ep}{\varepsilon}

\newcommand{\ov}{\overline}

\graphicspath{{paper-pics/}}

\begin{document}
\setlength{\parskip}{1mm}
\setlength{\oddsidemargin}{0.1in}
\setlength{\evensidemargin}{0.1in}
\lhead{}
\rhead{}
\rfoot{}
\lfoot{}

\begin{center}
{\bf \Large \noindent Analysis of new direct sampling indicators for far-field measurements }\\
\vspace{0.2in}
Isaac Harris  \\
Department of Mathematics\\
Purdue University \\
West Lafayette, IN 47907\\
E-mail: harri814@purdue.edu

\vspace{0.2in}

Andreas Kleefeld \\
Forschungszentrum J\"{u}lich GmbH \\
J\"{u}lich Supercomputing Centre \\
Wilhelm-Johnen-Stra{\ss}e, 52425 J\"{u}lich, Germany\\
Email: a.kleefeld@fz-juelich.de\\
\end{center}

\begin{abstract}
\noindent This article focuses on the analysis of three direct sampling indicators which can be used for recovering scatterers from the far-field pattern of time-harmonic acoustic measurements. These methods fall under the category of sampling methods where an indicator function is constructed using the far-field operator. Motivated by some recent work, we study the standard indicator using the far-field operator and two indicators derived from the factorization method. We show equivalence of two indicators previously studied  as well as propose a new indicator based on the Tikhonov regularization applied to the far-field equation for the factorization method. Finally, we give some numerical examples to show how the reconstructions compare to other direct sampling methods.  
\end{abstract}



\section{Introduction}\label{intro}
In this paper, the inverse problem of recovering the shape of an unknown obstacle from far-field measurements is considered. Such a problem arises in many physical applications where one needs to detect structures in a given medium. 
This has many applications such as medical imaging or detecting defects in complex material structures. 
An important question is if one can construct fast stable algorithms to reconstruct the obstacle with little to no a priori information. {\color{black} One way} to achieve this is to employ a {\it qualitative method} (otherwise known as non-iterative or direct methods) such as the linear sampling method (LSM) which was first proposed in \cite{CK} or the factorization method (FM) which has been  applied to acoustic, electromagnetic, and electrostatic problems in \cite{kirschbook} as well as the generalized linear sampling method (GLSM) introduced in \cite{GLSM} which combines theoretical elements of both the LSM and FM. These methods have been used to solve multiple inverse shape problems for elliptic \cite{armin}, hyperbolic \cite{HLM}, and parabolic \cite {lsm-heat} systems. The LSM studies the far-field equation and {\color{black}uses the fact that} the (regularized) solution to the far-field equation should become unbounded in the exterior of the obstacle. The FM gives that the far-field pattern of the fundamental solution to the Helmholtz equation is in the range of a self-adjoint compact operator defined by the far-field operator if and only if the sampling point is in the region of interest. Therefore, appealing to Picard's Criterion gives an indicator that becomes unbounded in the exterior of the obstacle by using the eigensystem of a known operator. There has also been some work in using numerical regularization techniques when applying the FM in \cite{pele}. Recently in \cite{cakoni-harris} and \cite{fm-crack} the FM has been extended to reconstructing defects in a known inhomogeneous media. The GLSM studies the minimizer of a functional to solve the far-field equation used in the LSM involving the far-field operator and it is shown that the penalty term in the functional becomes unbounded in the exterior of the obstacle. In the GLSM the penalty term can be defined by using the operator from the FM. All of these methods allow one to construct an indicator function $W(z)$ to recover the unknown obstacle where $W(z)$ is positive in the interior of the obstacle and  is (approximately) zero on the exterior of the obstacle. Therefore, to reconstruct the obstacle one can plot the computable function $W(z)$ in a region where the unknown scatterer is assumed to be located. Lately, there has been some interest in analyzing a so-called direct sampling method (DSM) using the theoretical framework used in the FM {\color{black}for far-field data (see for e.g. \cite{Liu} and \cite{dsm-fm}). These methods have also been studied for there applicability to reconstruct scatterers from near-field data \cite{rtm}.} The DSM recovers objects by constructing an indicator function by evaluating an inner-product with the measured far-field operator and the known far-field pattern of the fundamental solution to Helmholtz equation. This allows one to prove that these indicators are stable and computationally cheap to implement. Similar to the other sampling methods one wishes to show that the DSM indicator has a specific behavior outside the region of interest, where instead of blowing up outside the obstacle it takes smaller values in the exterior. 

Recently, in \cite{Liu} the DSM using the symmetric factorization of the far-field operator for time-harmonic acoustic data has been studied and the indicator function 
\[z \longmapsto \left| (F \phi_z , \phi_z )_{L^2(\mathbb{S})} \right|\]
is proposed, where $F$ is the far-field operator and $\phi_z$ is the far-field pattern of the fundamental solution {\color{black} located at $z$} to Helmholtz equation.
The analysis in this paper works for sound soft/hard, impedance obstacles as well as penetrable isotropic scatterers. In \cite{Liu} it is shown that the indicator is strictly positive and decays as $|z| \to \infty$. The decay is given by bounding the indicator by the Bessel functions evaluated at $|z-x|$ where $x$ is in the obstacle or on the boundary of the obstacle. It is also shown that the indicator is equivalent to reverse time migration \cite{rtm} and the orthogonal sampling method \cite{osm}. Similar analysis was employed in \cite{elastic} to validate the DSM for elastic waves as well as the study of its applicability for limited-aperture data. Unlike the other sampling methods the DSM does not require solving an ill-posed problem or minimizing a functional {\color{black}at each sampling point} which makes it computationally cheaper than its counterparts. Stability results can be given by simple calculations using basic analytical tools. This makes these methods advantageous to use in order to recover unknown obstacles. 

Here we propose the use of two indicator functions based on the FM. This research is motivated by the recent paper \cite{dsm-fm} where similar indicators are introduced and analyzed for the multi-static response matrix which is the discrete version of the far-field operator. The indicators in \cite{dsm-fm} are based on the theoretical setting of the FM and one of the proposed indicators is based on a Neumann series approximation to the FM equation. In this paper, we will analyze the indicator function which uses the FM operator and prove that the indicator is equivalent to DSM studied in \cite{Liu} and derive a new method based on the Tikhonov regularization. The main idea for deriving the new DSM is based on the factorization of the far-field operator given by $F$ which uses $|F|^{1/2} = (F^*F)^{1/4}$. 
If we let the unknown obstacle be denoted by $D$, then the FM gives that $ \phi_z  \in \mathrm{Range}(|F|^{1/2}) \iff z \in D$
and our method looks at an approximation of the Tikhonov regularized solution operator of the equation $|F|^{1/2} g_z =  \phi_z\,.$

The rest of the paper is structured as follows. First, we provide a precise problem statement in Section \ref{setup} and give some preliminary result. In Section \ref{DSM1} we analyze the first indicator which is defined by an inner-product involving $|F|^{1/2}$ and $\phi_z$ which was originally considered in \cite{dsm-fm}. The authors of \cite{dsm-fm} did not succeed in showing that the indicator is equivalent to the one proposed in \cite{Liu} and here we will show that the indicators are equivalent. In Section \ref{DSM2} we expand on the idea in \cite{dsm-fm} to use an approximate solution operator to define a new indicator function. In order to define our approximate solution operator we consider {\color{black}using} Tikhonov regularization. Section \ref{numerics} is devoted to giving some numerical examples of recovering scatterers in two and three dimensions as well as comparing the new indicators with the DSM indicators to show that our reconstructions can compete and in some cases outperform the DSM indicator considered in \cite{Liu}. In our experiments we see that the reconstructions take seconds to compute making these methods computationally cheap to implement and analytically rigorous. Lastly, a short summary is given in Section \ref{summary}.

\section{Problem statement}\label{setup}
In our analysis, we will consider the time-harmonic acoustic scattering problem for anisotropic media.  The LSM and FM have been applied to the inverse scattering problem of recovering the scatterer from the 
measured far-field pattern in \cite{lsm-aniso} and \cite{fm-aniso}, respectively. We now derive two DSMs for this inverse shape problem. 
To this end, we formulate the direct scattering problems under consideration in $\R^d$ for $d=2$ or $d=3$. The scatterer $D \subset \R^d$ may be made up of multiple simply connected components with $C^2$ boundary 
$\partial D$ having unit outward normal $\nu$. We consider the scattering by a plane incident wave $u^i(x, \hat{y}) = \mathrm{e}^{\mathrm{i} k {x}  \cdot \hat{y} }$ for a given incident direction $\hat{y} \in \R^d$ such that $| \hat{y} |=1$ and wave number $k>0$. This gives that the radiating scattered field $u^s(x, \hat{y})$ and the total field $u=u^s+u^i \in H^1_{loc}(\R^d)$ satisfy the boundary value problem
\begin{eqnarray}
\hspace{-1cm}\Delta u^s+k^2u^s = 0  \, \,  \text{ in } \, \, \R^d \setminus \overline{D} \quad &\text{and }& \quad  \grad \cdot A(x) \grad u +k^2 n(x) u = 0  \,\,\text{ in } \,\,  D \label{direct1} \\
 u=u^s+u^i \quad &\text{and } &\quad {\partial_{\nu_A} u}= {\partial_\nu} (u^s+u^i )     \,\, \textrm{ on }   \,\, \partial D  \label{direct2}
 \end{eqnarray}
along with the Sommerfeld radiation condition
\begin{eqnarray*}
\lim\limits_{r \rightarrow \infty} r^{(d-1)/2} \left( {\partial _r u^s} -\mathrm{i}k u^s \right)=0\,.
\end{eqnarray*}
Here the normal and conormal derivative on the boundary is given by ${\partial_{\nu} \varphi}=\nu \cdot   \nabla \varphi$ and ${\partial_{\nu_A} \varphi}=\nu \cdot A  \nabla \varphi$,
respectively, where $\varphi$ is a sufficiently regular function defined on $\partial D$. The Sommerfeld radiation condition is assumed to be satisfied uniformly with respect to $\hat x=x/r$, $r=|x|\,.$ 

The real-valued coefficient matrix $A(x) \in C^{1}(D,\R^{d \times d})$ is symmetric uniformly positive definite in $D$ and the scalar function $n(x) \in C(D)$ is real-valued, denotes the material parameters of the obstacle $D$ where we assume that $I-A$ and $1-n$ are supported in $D$. Under these assumptions we have that the scattering problem \eqref{direct1}--\eqref{direct2} with the radiation condition is well-posed for all $\hat{y}$. It can be shown that the radiating scattered field $u^s$ has the expansion 
\[u^s(x, \hat{y})=\gamma \frac{\mathrm{e}^{\mathrm{i}k|x|}}{|x|^{(d-1)/2}} \left\{ u^{\infty}(\hat{x},  \hat{y} ) + \mathcal{O} \left( \frac{1}{|x|}\right) \right\} \; \text{  as  } \;  |x| \to \infty
\]
where the constant 
\[\gamma = \frac{ \mathrm{e}^{\mathrm{i}\pi/4} }{ \sqrt{8 \pi k} } \,\,\, \text{in} \,\,\, \R^2 \quad \text{and} \quad  \gamma = \frac{1}{ 4\pi } \,\,\, \text{in} \,\,\, \R^3 \]
with $u^{\infty}(\hat{x},  \hat{y})$ being the {far-field pattern} depending on the incident direction $ \hat{y}$ and the observation direction $\hat x$. 
We now define the far-field operator $F:L^2(\mathbb{S}) \longmapsto  L^2(\mathbb{S})$
\begin{eqnarray}
(F g)(\hat{x})=\int_{\mathbb{S}} u^{\infty}(\hat{x}, \hat{y}  ) g(\hat{y} ) \, \mathrm{d}s(\hat{y} )\,, \quad \text{ where } \, \, \mathbb{S}=\text{unit sphere/circle.} \label{ffo}
\end{eqnarray}
Since the far-field pattern is analytic it is clear that $F$ is a compact operator. 
The two indicators we study are given by 
\[ \left( |F|^{1/2} \phi_z , \phi_z \right)_{L^2(\mathbb{S})} \quad \text{and} \quad \big\| {\color{black}P_{\alpha,\ep}}(|F|)\phi_z \big\|^2_{L^2(\mathbb{S})}\,, \quad \text{where} \quad \phi_z = \mathrm{e}^{ - \mathrm{i} k \hat{x}  \cdot z}\,.\]
Since $F^*F$ is a positive self-adjoint compact operator we can define $|F|^{1/2} = (F^*F)^{1/4}$ via the spectral decomposition.  Here, {\color{black}$P_{\alpha,\ep} (t)$ is a polynomial defined on the interval $[0, \|F\|]$ that approximates the Tikhonov regularized filter function with regularization parameter ${\alpha}$ and accuracy $\ep$} of the equation
\[|F|^{1/2} g_z =  \phi_z\quad \text{ for } z \in \R^d\,, \]
 see Section \ref{DSM2} for details. {\color{black} The merits of developing these new indicators is to theoretically justify the numerical investigation in \cite{dsm-fm} for the continuous (and discretized) setting. We also wish to analyze the new indicators to connect the recent developments for DSMs to earlier qualitative reconstructive algorithms. In our experiments computing the singular values/vector as well as constructing the polynomial $P_{\alpha,\ep}$ adds a negligible increase in computing time for a modest sized far-field matrix.} 

Before we begin, we will show that the indicator $\left| (F \phi_z , \phi_z )_{L^2(\mathbb{S})} \right|$ decays as dist$(z,D)$ increases for an anisotropic scatterer (see \cite{Liu} for other scatterers). It has been shown in \cite{CCH-book} that $F$ has the factorization $F=H^* T H$ such that  
\[H: L^2(\mathbb{S}) \to \big[L^2(D)\big]^{d+1} \quad \text{ is given by } \quad Hg = \left( \grad {v}_g ,  {v}_g \right)^{\top}\]
where $v_g$ is the Herglotz wave function defined as 
\[v_g(x) =  \int_{\mathbb{S}}  \mathrm{e}^{\mathrm{i} k {x}  \cdot \hat{y} } g(\hat{y} ) \, \mathrm{d}s(\hat{y} ) \,.\]
Here $H^*$ is the adjoint operator to $H$ and $T$ is a bounded linear operator from $\big[L^2(D)\big]^{d+1}$ to itself.  We have the identity 
\[(v_{\phi_z})(x) =  \int_{\mathbb{S}}  \mathrm{e}^{- \mathrm{i} k {(z - x)}  \cdot \hat{y} } \, \mathrm{d}s(\hat{y} ) =\left\{\begin{array}{lr} 2\pi J_0(k | x - z|) \, \, & \, \text{in} \, \, \R^2 \,,\\
 				&  \\
4\pi j_0(k | x - z|) & \,  \text{in} \,\,\R^3\,,
 \end{array} \right. \]
where $J_0$ is the zeroth order Bessel function of the first kind and $j_0$ the zeroth order spherical Bessel function of the first kind (see for e.g. \cite{Liu}). The factorization of the far-field operator implies that 
\begin{eqnarray*}
\left| (F \phi_z , \phi_z )_{L^2(\mathbb{S})} \right| &=& \left| ( TH \phi_z , H \phi_z )_{[L^2(D)]^{d+1}}  \right|  \\
									  &\leq& C\| H \phi_z \|_{[L^2(D)]^{d+1}}^2 \\
 									  &=& C \| v_{\phi_z} \|^2_{H^1(D)} \, .
\end{eqnarray*}
Recall that $J_0 (t)$ and its derivatives decay at a rate of $t^{-1/2}$ as $t \to \infty$. In $\R^3$ we have that $j_0 (t)$ and its derivatives decay rate of  $t^{-1}$ as $t \to \infty$.
Therefore, the above inequality and the decay of the Bessel functions gives that   
\begin{theorem}\label{decay}
For all $z \in \R^d \setminus \ov{D}$ 
$$ \left| (F \phi_z , \phi_z )_{L^2(\mathbb{S})} \right|  =  \mathcal{O} \left(\mathrm{dist}(z,D)^{1-d}\right) {\color{black} \quad \text{ as } \quad\mathrm{dist}(z,D) \to \infty.}$$
\end{theorem}
We then see that the indicator will decay as $z$ moves away from the scatterer $D$. Note that Theorem \ref{decay} is valid for the case when $A$ and $n$ are complex-valued functions.

\section{A factorization based direct sampling method} \label{DSM1}

In this section, we will study the indicator using the operator $|F|^{1/2} = (F^*F)^{1/4}$. It is well known that one can uniquely recover the scatterer $D$ using the far-field pattern (see Chapter 6 of  \cite{cakoni-colton}). The purpose of analyzing the DSM is to derive stable and computationally simple reconstruction algorithms. The stability of the DSMs proposed in this section have been studied in \cite{Liu} and \cite{dsm-fm} (for the discretized case). We will show that the corresponding indicator functions 
\[  W_{\text{FDSM}}(z)=\left( |F|^{1/2} \phi_z , \phi_z \right)_{L^2(\mathbb{S})}  \quad \text{ and } \quad W_{\text{DSM}}(z) = \left| (F \phi_z , \phi_z )_{L^2(\mathbb{S})} \right| \]
with $\phi_z = \mathrm{e}^{ - \mathrm{i} k \hat{x}  \cdot z}$ are equivalent.  {\color{black}This crucial theoretical result is needed to prove the validity of the new indicator $W_{\text{{FDSM}}}(z)$  which is not established in \cite{dsm-fm}.} Here $W_{\text{{FDSM}}}(z)$ is the DSM based on the FM and $W_{\text{{DSM}}}(z)$ is the standard DSM studied in \cite{Liu}. To this end, we will bound $W_{\text{DSM}}(z) $ from above and below by our new indicator $W_{\text{{FDSM}}}(z)$. Since $W_{\text{{DSM}}}(z)$ decays as the dist$(z,D)$ increases {\color{black} this suggests that} we can plot either function to recover the scatterer. 

To begin, we need a few results for the far-field operator defined in \eqref{ffo} that is associated with \eqref{direct1}--\eqref{direct2} along with the radiation condition. We introduce the scattering operator associated with this direct scattering problem \eqref{direct1}--\eqref{direct2}. The  scattering operator ${ S}:L^2(\mathbb{S}) \to L^2(\mathbb{S})$  is defined by 
\[{S}=I+2 \mathrm{i} k |\gamma|^2 F\,.\]
Since $A(x)$ and $n(x)$ are assumed to be real-valued, we have that the scattering operator is unitary, ${S} {S}^*={S}^*{S}=I$ (see for e.g. \cite{cakoni-colton}). Here $I$ denotes the identity operator on $L^2(\mathbb{S}) $. This implies that the corresponding far-field operator $F:L^2(\mathbb{S}) \to L^2(\mathbb{S})$ is normal and compact. We can now conclude that $F$ has an orthonormal eigenvalue decomposition $(\lambda_j , \psi_j) \in \C \times L^2(\mathbb{S})$ such that 
\[ Fg = \sum\limits_{j=1}^{\infty} \lambda_j  (g,\psi_j)_{_{L^2(\mathbb{S})}} \psi_j \quad \text{for all } \quad  g \in L^2(\mathbb{S})\,.\]
Since $F$ is a compact operator we have that $|\lambda_j| \to 0$ as $j \to \infty$. Provided that the wave number $k$ is not an interior transmission eigenvalue we have that $F$ is injective with a dense range (see \cite{cakoni-colton,CCH-book}). This implies that  $|\lambda_j| \neq 0$ for all $j$ and that the set $\{ \psi_j\}$ is a complete orthonormal set in $L^2(\mathbb{S})$. Following the results of \cite[Chapter 1]{kirschbook} it can be shown that the far-field operator has the factorization 
\begin{eqnarray}
F = |F|^{1/2} Q |F|^{1/2} \label{factor}
\end{eqnarray}
where we use the eigensystem to define the operators
\[ |F|^{1/2} g =  \sum\limits_{j=1}^{\infty}\sqrt{ |\lambda_j |} (g,\psi_j)_{_{L^2(\mathbb{S})}} \psi_j  \quad \text{ and } \quad Q g =  \sum\limits_{j=1}^{\infty} \frac{\lambda_j}{ |\lambda_j |} (g,\psi_j)_{_{L^2(\mathbb{S})}} \psi_j\,.\]
It is clear that $|F|^{1/2}$ and $Q$  are {\color{black} bounded linear operators that map $L^2(\mathbb{S})$ into itself}. Since $F$ is injective we have that $|F|^{1/2}$ is a positive self-adjoint compact operator. 
Since the scattering operator is unitary we can conclude that for all $g \in L^2(\mathbb{S})$
\begin{eqnarray}
 \big| (Qg,g)_{{L^2(\mathbb{S})}} \big| \geq \mu \| g \|^2_{{L^2(\mathbb{S})}} \quad \text{ and } \quad  \|Qg\|_{{L^2(\mathbb{S})}} \leq \| g\|_{{L^2(\mathbb{S})}} \label{Q}
\end{eqnarray}
for some positive constant $\mu$ (see for e.g. Theorem 7.29 in \cite{cakoni-colton} for details). 

In \cite{dsm-fm} the authors did not succeed in bounding $W_{\text{{FDSM}}}(z)$ above by $W_{\text{{DSM}}}(z)$ for the discretized far-field operator which is what is needed for showing the equivalence of these two indicators. From the factorization of the far-field operator \eqref{factor} and the above inequality \eqref{Q}  we have the following estimates  
\begin{eqnarray*}
\left| (F \phi_z , \phi_z )_{L^2(\mathbb{S})} \right| &=&  \left| \left(Q |F|^{1/2} \phi_z , |F|^{1/2} \phi_z \right)_{L^2(\mathbb{S})}\right| \\
									  & \geq & \mu \big\| |F|^{1/2} \phi_z \big\|^2_{L^2(\mathbb{S})}\\
									  &= &\mu \sup\limits_{\| \varphi \|_{L^2(\mathbb{S})} = 1}  \left|  \left( |F|^{1/2} \phi_z , \varphi \right)_{L^2(\mathbb{S})}\right|^2\\
									  & \geq& \frac{\mu}{2^{d-1} \pi}   \left|  \left( |F|^{1/2} \phi_z , \phi_z \right)_{L^2(\mathbb{S})}\right|^2  
\end{eqnarray*}
since
\[\| \phi_z \|^2_{L^2(\mathbb{S})} = 2^{d-1} \pi\]
holds.
The above estimate gives that $W^2_{\text{{FDSM}}}(z)$ will have the same decay as dist$(z,D)$ increases which is the critical piece missing in the manuscript \cite{dsm-fm}.  We have just about all we need to prove that the two indicators are equivalent. In order to prove the upper bound we need one last theoretical lemma for positive operators on a Hilbert space to prove the equivalence. 
 \begin{lemma}\label{positive}
 Let $T : \mathcal{V} \to  \mathcal{V}$ be a bounded positive operator on a Hilbert space $\mathcal{V}$ then we have that 
 \[ \|T v \|^2_{ \mathcal{V}} \leq \|T\| (Tv,v)_{\mathcal{V}} \quad \text{ for all } \, \, v\in \mathcal{V}\,.\]
 \end{lemma}
 \begin{proof}
For the proof we refer the reader to \cite[Lemma 2.1]{positive-lemma}.
 \end{proof}

Notice that since $F$ is injective we have that $|F|^{1/2}$ is a positive operator acting on the Hilbert space $L^2(\mathbb{S})$. Using the above lemma we have that 
\begin{eqnarray*}
\left| (F \phi_z , \phi_z )_{L^2(\mathbb{S})} \right| &=&  \left| \left(Q |F|^{1/2} \phi_z , |F|^{1/2} \phi_z \right)_{L^2(\mathbb{S})}\right| \\
									  & \leq &  \big\| |F|^{1/2} \phi_z \big\|^2_{L^2(\mathbb{S})}\\
									  &\leq &  \left\| |F|^{1/2} \right\|  \left( |F|^{1/2} \phi_z , \phi_z \right)_{L^2(\mathbb{S})}\,.  
\end{eqnarray*}
Thus, we have proven that $W_{\text{{DSM}}}(z)$ is bounded above by $W_{\text{{FDSM}}}(z)$ which implies that they are equivalent in the sense that 
\[ \frac{\mu}{2^{d-1} \pi}   W^2_{\text{{FDSM}}}(z) \leq W_{\text{{DSM}}}(z) \leq  \left\| |F|^{1/2} \right\| W_{\text{{FDSM}}}(z)\,.\]
Therefore, we have that the following result. 
\begin{theorem}
There {\color{black} exist} two positive constants $c_1$ and $c_2$ such that 
\[c_1 \left|  \left( |F|^{1/2} \phi_z , \phi_z \right)_{L^2(\mathbb{S})}\right|^2 \leq  \left| (F \phi_z , \phi_z )_{L^2(\mathbb{S})} \right|  \leq c_2 \left( |F|^{1/2} \phi_z , \phi_z \right)_{L^2(\mathbb{S})}\,.\]
\end{theorem}

Since we have shown in Theorem \ref{decay} that $W_{\text{DSM}}(z)$ is bounded by the squared $H^1(D)$ norm of the zeroth order Bessel function in $\R^d$  for $d=2$ or $3$ we have that $W_{\text{{DSM}}}(z) = \mathcal{O}\left(\mathrm{dist}(z,D)^{1-d}\right)\,$ {\color{black} as $\,\mathrm{dist}(z,D) \to \infty$} when the sampling point $z \in \R^d \setminus \ov{D}$ which gives the following result. 
\begin{theorem}\label{decay-fdsm}
For all $z \in \R^d \setminus \ov{D}$ 
\[\left|  \left( |F|^{1/2} \phi_z , \phi_z \right)_{L^2(\mathbb{S})}\right|^2 =  \mathcal{O} \left(\mathrm{dist}(z,D)^{1-d}\right) {\color{black} \quad \text{ as } \quad\mathrm{dist}(z,D) \to \infty.} \]
\end{theorem}

Notice that since $|F|^{1/2}$ is a positive operator the function $W_{\text{{FDSM}}}(z)$ is strictly positive for all $z$ and as the sampling point moves away from the boundary $W_{\text{{FDSM}}}(z)$ decays. Using the eigensystem for $F$ we have that 
\[  W_{\text{{FDSM}}}(z) =   \sum\limits_{j=1}^{\infty}\sqrt{ |\lambda_j |} \left| (\phi_z ,\psi_j)_{_{L^2(\mathbb{S})}}\right|^2 \]
compared to the indicator given by the FM defined as 
\[ z \longmapsto  \sum\limits_{j=1}^{\infty} \frac{1}{ |\lambda_j |} \left| (\phi_z ,\psi_j)_{_{L^2(\mathbb{S})}}\right|^2\,. \]
Since the eigenvalues tend to zero rapidly one should avoid dividing by them in practice. The function $W_{\text{{FDSM}}}(z)$ is computed by multiplying by the square roots of the eigenvalues which will be more stable in the presence of errors in the measured far-field data. {\color{black} One drawback is that Theorems \ref{decay} and \ref{decay-fdsm} only give that the new indicators decay as the sampling point moves away from the scatterer which may result in low contrast reconstructions. In \cite{Liu} it is seen that this can be overcome by raising the new indicators to the power $p>1$ to sharpen the resolution. To recover the scatterer one can take a level curve of the indicator $W^p$. }        

Notice that the analysis in this section only requires that the far-field operator is injective with dense range and is normal which implies the orthonormal eigenvalue decomposition. The fact that the corresponding scattering operator is unitary is the key component to the analysis in this section. Therefore, this equivalence of the indicators holds for any scattering problem where the corresponding scattering operator is unitary. This is true for the inverse obstacle scattering problem where the scattered field solves 
\[ \Delta u^s+k^2u^s = 0  \, \,  \text{ in } \, \, \R^d \setminus \overline{D} \quad \text{ and } \quad \mathcal{B}(u^s)=- \mathcal{B}(u^i) \,\, \text{ on } \,\,  \partial D\]
along with the Sommerfeld radiation condition. Here the boundary operator is given by 
\[\mathcal{B}(\varphi)=\varphi \quad \text{ or } \quad  \mathcal{B}(\varphi)=\partial_{\nu} \varphi + \gamma(x) \varphi\]
where $\gamma \in L^{\infty}(\partial D)$ is non-negative. This implies that $W_{\text{{FDSM}}}(z)$ can be used to recover sound soft/sound hard, isotropic, and impedance type scatterers. 

\section{A Tikhonov regularization based direct sampling method} \label{DSM2}
The operator $|F|^{1/2}$ has been used to recover the scattering objects in previous studies where one solves the ill-posed equation
\begin{eqnarray}
|F|^{1/2} g_z =\phi_z \quad \text{ for } z \in \R^d \label{fm-equ}
\end{eqnarray}
which is solvable if and only if the sampling point $z \in D$. One of the main ideas proposed in \cite{dsm-fm} is to derive an approximate solution operator to the above equation and use the approximate solution operator to define a DSM. In \cite{dsm-fm} the authors approximate the solution operator using a Neumann series. Using a Neumann series to approximate the solution operator amounts to constructing a polynomial that when applied to the operator acts as the solution operator of \eqref{fm-equ}. The main idea we exploit is to construct a polynomial that when evaluated at the operator $|F|$ acts as an approximate solution operator for \eqref{fm-equ}. Here we propose approximating the solution operator using Tikhonov regularization, which is commonly used in the literature to solve \eqref{fm-equ} (see for e.g. \cite{fm-gbc}). 

The analysis in this section again appeals to the eigenvalue decomposition. Now, recall that the orthonormal eigenvalue decomposition of the injective far-field operator $F$ is given by $(\lambda_j , \psi_j) \in \C \setminus \{0\} \times L^2(\mathbb{S})$. Therefore, we can define $|F|^p$ for $p>0$ by
\[ |F|^{p} g =  \sum\limits_{j=1}^{\infty}{ |\lambda_j |}^p (g,\psi_j)_{_{L^2(\mathbb{S})}} \psi_j   \]
where the set $\{ \psi_j\}$ is an orthonormal basis in $L^2(\mathbb{S})$.
Note that the Tikhonov regularized solution of \eqref{fm-equ} will be denoted $g_z^{\alpha}$ and is the unique minimizer of the functional 
\[ \big\| |F|^{1/2}   g_z^{\alpha} - \phi_z  \big\|^2_{L^2(\mathbb{S})} +\,  \alpha \, \big\| g_z^{\alpha} \big\|^2_{L^2(\mathbb{S})}\]
where $\alpha > 0$ is the regularization parameter. Simple calculations give that the minimizer $g_z^{\alpha}$ satisfies the equation 
\[ \alpha g_z^{\alpha}  +  |F|  g_z^{\alpha}  = |F|^{1/2} \phi_z \quad \text{ for } z \in \R^d\]
and using the eigenvalue decomposition we have that 
\[ g_z^{\alpha}  =  \sum\limits_{j=1}^{\infty} \frac{ \sqrt{ |\lambda_j |}}{\alpha +  |\lambda_j |} (\phi_z,\psi_j)_{_{L^2(\mathbb{S})}} \psi_j\,.\]
We now define the function $\Gamma_{\alpha} (t)= \frac{ \sqrt{t}}{\alpha + t}$ which is continuous on the interval $[0, \|F\|]$ and we have the solution operator for the Tikhonov regularization of \eqref{fm-equ} given by the mapping 
\begin{eqnarray}
\phi_z \longmapsto  \sum\limits_{j=1}^{\infty} \Gamma_{\alpha} ( |\lambda_j |) (\phi_z,\psi_j)_{_{L^2(\mathbb{S})}} \psi_j\,. \label{tik-solu}
\end{eqnarray}
In general, the regularization parameter  $\alpha$ is taken to be small. The parameter $\alpha$ is chosen to be a fixed but small parameter throughout all the calculations in our experiments provided in Section \ref{numerics}.

In order to approximate the solution operator in \eqref{tik-solu} we exploit the fact that for all $\alpha >0$ the function $\Gamma_{\alpha}(t)$ is actually continuous for all $t\geq0$ and therefore we have that for every $\ep >0$ there is a polynomial ${\color{red} P_{\alpha,\ep}(t)}$ such that 
\begin{eqnarray}
\| {\color{black} P_{\alpha,\ep}}(t) -  \Gamma_{\alpha} (t) \|_{L^{\infty}(0, \|F\|) } < \ep. \label{poly}
\end{eqnarray}
{\color{black}The approximation of the solution operator is now defined by} 
\[{\color{black} P_{\alpha,\ep}}(|F|) \phi_z =\sum\limits_{j=1}^{\infty} {\color{black} P_{\alpha,\ep}}( |\lambda_j |) (\phi_z,\psi_j)_{_{L^2(\mathbb{S})}} \psi_j .\]
We have defined the polynomial of the operator $|F|$ via the eigenvalue decomposition as is commonly done in Linear Algebra. The Tikhonov indicator we propose in this section for a fixed $\alpha$ positive is defined as the function 
\[W_{\text{TDSM}}(z) = \big\| {\color{black} P_{\alpha,\ep}}(|F|)\phi_z \big\|^2_{L^2(\mathbb{S})} \quad \text{with} \quad \left\| {\color{black} P_{\alpha,\ep}}( t) - \frac{ \sqrt{t}}{\alpha + t} \right\|_{L^{\infty}(0, \|F\|) } \approx 0\]
where ${\color{black} P_{\alpha,\ep}}(t)$ is a polynomial. By definition of ${\color{black} P_{\alpha,\ep}}(|F|)\phi_z$ we have that 
\[\big\| {\color{black} P_{\alpha,\ep}}(|F|)\phi_z \big\|^2_{L^2(\mathbb{S})}=\sum\limits_{j=1}^{\infty} {\color{black} P^2_{\alpha,\ep}}( |\lambda_j |) \left| (\phi_z ,\psi_j)_{_{L^2(\mathbb{S})}}\right|^2\]
by appealing to the fact that $\{ \psi_j\}$ is an orthonormal set in $L^2(\mathbb{S})\,.$    

The goal now is to see how the new Tikhonov indicator $W_{\text{TDSM}}(z)$ compares to the indicators studied in the previous sections. 
To do so, we assume that the {\color{black} regularization parameter is known and fixed} and note that for the polynomial \textcolor{black}{$P_{\alpha,\ep}(t)$} satisfying \eqref{poly} we have that 
\[ {\color{black} P^2_{\alpha,\ep}}( |\lambda_j |) \leq   \Gamma^2_{\alpha} ( |\lambda_j |) +2 \ep  \Gamma_{\alpha} ( |\lambda_j |) + \ep^2 \quad \text{ for all } \quad \ep>0\quad \text{and}\quad j\in \mathbb{N} \,.\]
Since $\Gamma_{\alpha}(t)$ is continuous for $t\geq 0$, we let 
\[C_{\alpha} =  \left\|\Gamma_{\alpha}(t)  \right\|_{L^{\infty}(0, \|F\|) } < \infty\]
which only depends on $\alpha$ and $\|F\|$. Using basic calculus one can easily show that 
\[C_{\alpha} = \max \left\{ \frac{1}{2 \sqrt{\alpha} } , \frac{ \sqrt{\|F\|} }{\alpha +\|F\|} \right\}\,.\]
Now assume that $0<\ep <1$ and estimate 
\begin{eqnarray*}
\big\| {\color{black} P_{\alpha,\ep}}(|F|)\phi_z \big\|^2_{L^2(\mathbb{S})}  &\leq& \sum\limits_{j=1}^{\infty}  \Gamma^2_{\alpha}( |\lambda_j |) \left| (\phi_z ,\psi_j)_{_{L^2(\mathbb{S})}}\right|^2  \\
                                    & + &2 \ep \sum\limits_{j=1}^{\infty}  \Gamma_{\alpha} ( |\lambda_j |) \left| (\phi_z ,\psi_j)_{_{L^2(\mathbb{S})}}\right|^2 + \ep^2 \sum\limits_{j=1}^{\infty} \left| (\phi_z ,\psi_j)_{_{L^2(\mathbb{S})}}\right|^2 \\
								   &\leq& \sum\limits_{j=1}^{\infty}  \Gamma^2_{\alpha}( |\lambda_j |) \left| (\phi_z ,\psi_j)_{_{L^2(\mathbb{S})}}\right|^2 + \left(2 \ep C_\alpha + \ep^2\right)  \| \phi_z \|^2_{L^2(\mathbb{S})} \\
								   &=& \sum\limits_{j=1}^{\infty}  \Gamma^2_{\alpha}( |\lambda_j |) \left| (\phi_z ,\psi_j)_{_{L^2(\mathbb{S})}}\right|^2 +  2^{d-1} \pi \left(2 \ep C_\alpha + \ep^2\right)\,. 
\end{eqnarray*}
Notice that by definition $\Gamma^2_{\alpha}( |\lambda_j |) \leq |\lambda_j | / \alpha^2$ which gives that 
\[\big\| {\color{black} P_{\alpha,\ep}}(|F|)\phi_z \big\|^2_{L^2(\mathbb{S})}  \leq  \frac{1}{\alpha^2} \sum\limits_{j=1}^{\infty}  |\lambda_j | \left| (\phi_z ,\psi_j)_{_{L^2(\mathbb{S})}}\right|^2 + C(\alpha , d) \ep\]
where $C(\alpha , d)$ is a positive constant depending on the regularization parameter and the dimension. By definition of $|F|^{1/2}$ it is clear that 
\[\big\| {\color{black} P_{\alpha,\ep}}(|F|)\phi_z \big\|^2_{L^2(\mathbb{S})}  \leq  \frac{1}{\alpha^2}  \big\| |F|^{1/2} \phi_z \big\|^2_{L^2(\mathbb{S})}  + C(\alpha , d) \ep\,.\]
Therefore, by appealing to estimates in Section \ref{DSM1} we have the following result. 
\begin{theorem} \label{TDMS}
For all fixed $\alpha>0$  we have 
\[\big\| {\color{black} P_{\alpha,\ep}}(|F|)\phi_z \big\|^2_{L^2(\mathbb{S})}  \leq  \frac{1}{\mu \alpha^2} \left| (F \phi_z , \phi_z )_{L^2(\mathbb{S})} \right| + \mathcal{O}( \ep ) \quad {\color{black} \text{ as } \quad \ep \to 0}\]
where the constant $\mu$ is defined in \eqref{Q} and the polynomial ${\color{black} P_{\alpha,\ep}}(t)$ satisfies \eqref{poly}. 
\end{theorem}
The above theorem {\color{black}suggests} that the Tikhonov indicator $W_{\text{TDSM}}(z)$ {\color{black}decays at least as fast as} the standard DSM indicator $W_{\text{DSM}}(z)$ {\color{black}(up to order $\ep$)} and therefore we expect that $W_{\text{TDSM}}(z)$ has to approximately decay at a rate of $\mathrm{dist}(z,D)^{1-d}$ by Theorem \ref{decay} which implies that   
\[\big\| {\color{black} P_{\alpha,\ep}}(|F|)\phi_z \big\|^2_{L^2(\mathbb{S})}  \leq  C \mathrm{dist}(z,D)^{1-d} + \mathcal{O}( \ep ) \quad \text{for } \,\, z \in \R^d \setminus \ov{D}\, \]
{\color{black} as $\ep \to 0$}. Notice that again the analysis of this section requires that the far-field operator is injective with dense range and is normal which holds for sound soft/sound hard, isotropic, and impedance type scatterers. Therefore, the above estimate for the indicator $W_{\text{TDSM}}(z)$ holds for these scatterers as well.

\section{Numerical validation of the indicators} \label{numerics}
In this section, we provide some numerical examples for the three indicator functions studied in the previous sections. To this end, we will use synthetic far-field data in our numerical experiments. All of our experiments are done with MATLAB 2018a on an iMac with a 4.2 GHz Intel Core i7 processor with 8GB of memory. We will denote the discretized far-field operator as 
\[{\bf F} = \big[ u^{\infty}(\hat{x}_i, \hat{y}_j ) \big]_{i,j = 1}^M \quad \text{ where } \, \, \hat{x}_i, \hat{y}_j  \,\, \text{ are in } \,\, \mathbb{S}=\text{unit circle/sphere}\]
where $M$ is the number of incident and observation directions. Here the incident and observation directions are uniformly spaced on the unit circle/sphere.  We give examples with random noise added to the simulated data for $u^{\infty}(\hat{x}_i, \hat{y}_j )$. The random noise level is given by $\delta$ where the noise is added to the matrix such that 
\[{\bf F}_{\delta} = \left[u^{\infty}(\hat{x}_i, \hat{y}_j ) \left( 1 +\delta E_{i,j} \right) \right]_{i,j=1}^{M}\,,\]
 {with the complex-valued random matrix} $\mathbf{E}$ {satisfying } $\| \mathbf{E} \|_2 =1\,.$ Here $\|\, \cdotp \|_2$ denotes the spectral norm of a given matrix.

To evaluate our indicators we also define the vector 
\[{\bf \phi}_z = [\mathrm{e}^{ - \mathrm{i} k \hat{x}_j \cdot z}]_{j=1}^M\,.\]
Therefore, we have that the DSM indicator $ \left| (F \phi_z , \phi_z )_{L^2(\mathbb{S})} \right| $ is computed {numerically} by 
\[{\widetilde{W}_{\text{DSM}}(z) } = \big| {{\bf \phi}^*_z}{\bf F}_{\delta}   {\bf \phi}_z \big|\]
where the $^*$ denotes the conjugate transpose and {$\widetilde{W}_{\text{DSM}}(z)$ denotes the approximation of $W_{\text{DSM}}(z)$}. For the indicator function based on the FM we need to {numerically} compute the absolute value of ${\bf F}_{\delta}$ to the half power. To do so, we  let 
\[{\bf F}_{\delta} = {\bf U}{\bf S}{\bf V}^*\]
be the singular value decomposition of the matrix. Using the identity $|{\bf F}_{\delta} |^{1/2} = ({\bf F}_{\delta} ^*{\bf F}_{\delta} )^{1/4}$ we have that $|{\bf F}_{\delta} |^{1/2} = {\bf V}{\bf S}^{1/2}{\bf V}^*$. Therefore, we can numerically {approximate} the indicator $\left( |F|^{1/2} \phi_z , \phi_z \right)_{L^2(\mathbb{S})}$ using the singular value decomposition such that 
\[{\widetilde{W}_{\text{FDSM}}(z)} =   \sum\limits_{j=1}^{M} \sqrt{ s_j }  \big| {{\bf \phi}^*_z}  {\bf v}_j \big|^2\]
where $(s_j , {\bf v}_j ) \in \R_{> 0} \times \C^M$ are the singular values and right singular vectors of ${\bf F}_{\delta}\,.$ 

Now for the indicator based on the Tikhonov regularization we need to construct a polynomial $P_{\alpha}(t)$ such that for all $t \in [0 , \| {\bf F}_{\delta} \|]$ approximates the function $\Gamma_{\alpha} (t)$ defined in the previous section. Due to the high condition number for polynomial interpolation we notice that since $\Gamma_{\alpha} (0) = 0$ we construct an approximating polynomial that has zero as a root. 
This gives one coefficient less to compute. In our experiments we compute ${\color{black} P_{\alpha,\ep}}(t)$, where 
\[{\color{black} P_{\alpha,\ep}}(t) =    \sum\limits_{k=1}^{3} c_k t^k \quad \text{ such that } \quad{\color{black} P_{\alpha,\ep}}(t_\ell) = \frac{ \sqrt{t_\ell}}{\alpha + t_\ell} \]
with $\ell = 1, \ldots , 10$ and $t_\ell$ are equally spaced point in the interval $[0 , \| {\bf F}_{\delta} \|]$. One normally attempts to pick an optimal regularization parameter $\alpha$ via Morozov's discrepancy principle. Here, we only require $\alpha$ to be positive for Theorem \ref{TDMS} to hold which is all one needs to establish the decay of the indicator function. 
In general, the parameter $\alpha$ is taken to be small such that $\Gamma_{\alpha} (t) \approx 1/\sqrt{t}$ holds in order to make the mapping \eqref{tik-solu} the approximate solution of \eqref{fm-equ}. 
To this end, we fix $\alpha = 10^{-2}$ ad hoc in our experiments and use a spectral cut-off to compute the coefficients $c_k$ where the cut-off parameter is fixed to be $10^{-8}$ in all the examples. Once ${\color{black} P_{\alpha,\ep}}(t)$ is computed, we can numerically approximate the indicator $\big\| {\color{black} P_{\alpha,\ep}}(|F|)\phi_z \big\|^2_{L^2(\mathbb{S})}$ such that 
\[{\widetilde{W}_{\text{TDSM}}(z)} =   \sum\limits_{j=1}^{M} {\color{black} P^2_{\alpha,\ep}}({ s_j })  \big| {{\bf \phi}^*_z}  {\bf v}_j \big|^2\]
where we have again used the singular values $ s_j$ and right singular vectors ${\bf v}_j$. Notice that ${\color{black} P_{\alpha,\ep}}(t)$ need only be constructed once and is continuously used to evaluate the indicator. 
Two interesting questions arise about the implementation of this method: how to construct the polynomial ${\color{black} P_{\alpha,\ep}}(t)$ and how should one choose the regularization parameter? These two questions are not discussed here, but could lead to an interesting numerical and analytical investigation in the future.

Here we normalize the three indicator functions by dividing by their $L^{\infty}$ norms. In our experiments we see that the functions will be approximately one near the boundary of the scatterer and decay as the sampling point $z$ moves away from the {\color{black}scatterer}. Therefore, one can choose a level curve to recover the boundary of the scatterer which should be taken between $0.8$ and $0.9$ for the 2D case and between $0.6$ and $0.8$ for the 3D case from our experiments. 
In the following subsections we see that in some cases the indicators $\widetilde{W}_{\text{FDSM}}(z)$  and $\widetilde{W}_{\text{TDSM}}(z)$ seem to give a better contrast in the reconstructions which gives a more detailed approximation of the shape of the scatterer.

\subsection{Numerical results in two dimensions}
Here we consider reconstructing small isotropic scatterers (i.e. $A=I$) in $\R^2$. We have that the scattered field is given by the solution to the Lippmann-Schwinger integral equation (see for e.g. \cite{coltonkress})
\[u^s(x,\hat{y}) = \frac{\mathrm{i} }{4} k^2 \int_{D} \left( n(w) - 1\right) H^{(1)}_0 (k | x-w |)  u(w ,\hat{y}) \, \mathrm{d}w\]
where $H^{(1)}_0$ is the first kind Hankel function of order zero and $u$ is the total field for \eqref{direct1}--\eqref{direct2}. We consider the boundary of the domain $\partial D = r(\theta) \left(\cos (\theta), \sin(\theta) \right)$ where $r(\theta)$ is given by 
\begin{eqnarray*}
r(\theta) &=&\frac{1}{5} \left(2+\frac{3}{10}\cos(3\theta) \right) \quad \text{pear-shaped domain,} \\
r(\theta) &=& \frac{1}{5}\left(2+\frac{3}{10}\cos(5\theta) \right)\quad \text{star-shaped domain,} \\
r(\theta) &=& \frac{2}{5}\sqrt{\frac{1}{2}\sin(\theta)^2+\frac{1}{10}\cos(\theta)^2} \quad \text{peanut-shaped domain.} 
\end{eqnarray*}

\noindent{See Figure \ref{ref-domain2d} for a plot of the pear-shaped, star-shaped and peanut-shaped domain.} 

\begin{figure}
\includegraphics[width=15cm]{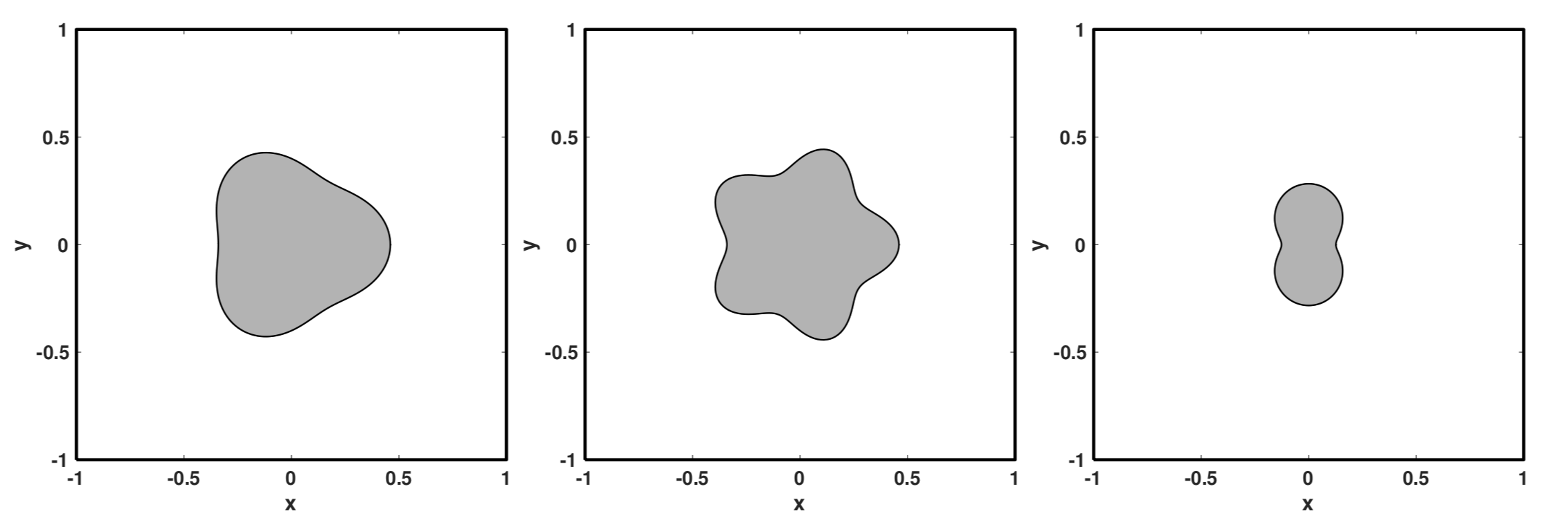}
\caption{{Plots of the pear-shaped, star-shaped, and peanut-shaped scatterers.}}
\label{ref-domain2d}
\end{figure}

Since the domains under consideration are small it is well known that the Born approximation, which is the first term in the Neumann series solution to the Lippmann-Schwinger integral equation, approximates the scattered field. It is given by 
\[u^{s}({x},\hat{y}) \approx \frac{\mathrm{i} }{4} k^2 \int_{D} \left( n(w) - 1\right) H^{(1)}_0 (k | x-w |)   \mathrm{e}^{\mathrm{i} k {w}  \cdot \hat{y} } \, \mathrm{d}w\,.\]
This gives that the far-field pattern can be approximated by 
\[u^{\infty}(\hat{x},\hat{y}) \approx k^2 \int_{D} \left( n(w) - 1\right) \mathrm{e}^{\mathrm{i} k {w}  \cdot (\hat{y} -  \hat{x})} \, \mathrm{d}w\]
which is a good approximation for the far-field pattern since $|D| \ll 1$. To evaluate the integral we use the built-in numerical 2D integrator `integral2' in MATLAB. 

In Figures \ref{recon1}--\ref{recon3} we fix the wave number $k=10$ and the refractive index $n=1/2$. We add $5\%$ random noise to the approximated far-field pattern which corresponds to $\delta =0.05$. 
There are $M=32$ uniformly spaced incident and observation directions given by $\hat{y}_j=\hat{x}_j=\left(\cos(\theta_j), \sin(\theta_j) \right)$ where $\theta_j$ are 
uniformly spaced points in $[0,2\pi)$. The sampling region is given by $[-1 , 1] \times [-1 , 1]$ where the sampling points are taken to be $100 \times 100$ equally spaced points in the sampling region. 
Computing the three indicator functions takes roughly three seconds for each of the domains under consideration. In Figures \ref{recon1}--\ref{recon3} we plot the three 
indicator functions studied in the previous section. We see in Figures \ref{recon1}--\ref{recon3} that the two indicators $\widetilde{W}_{\text{FDSM}}(z)$ and $\widetilde{W}_{\text{TDSM}}(z)$ give comparable results to what is given by $\widetilde{W}_{\text{DSM}}(z)$. 
\begin{figure}
\includegraphics[width=15cm]{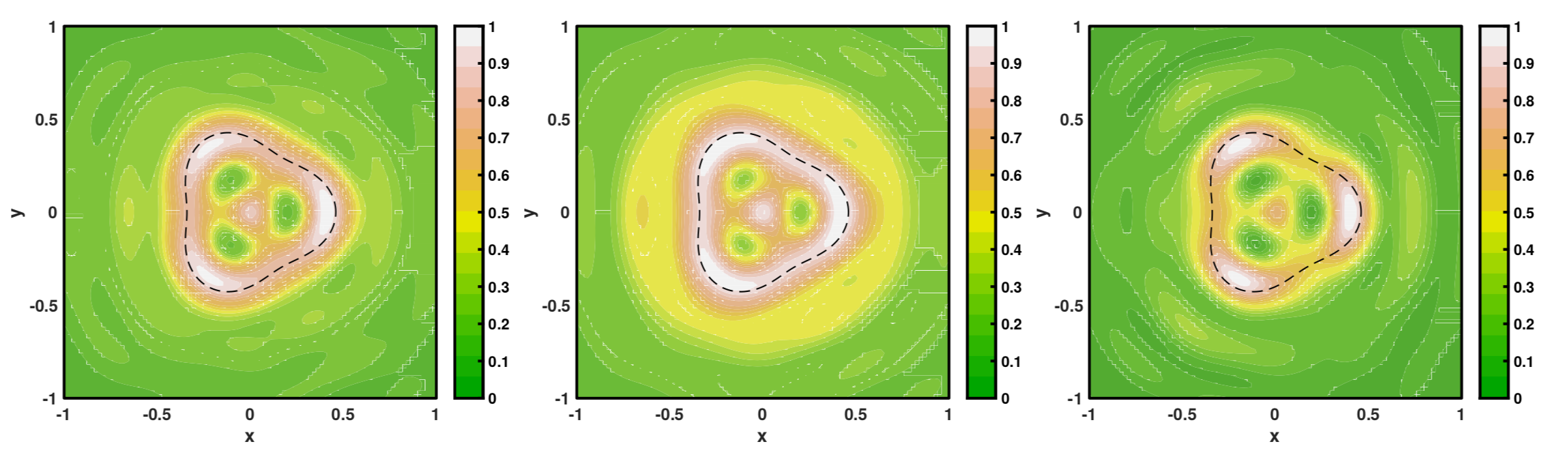}
\caption{{Reconstructions of the pear-shaped domain using the three indicator functions  $\widetilde{W}_{\text{DSM}}(z)$, $\widetilde{W}_{\text{FDSM}}(z)$, and $\widetilde{W}_{\text{TDSM}}(z)$ with noise level $\delta=5\%$.}}
\label{recon1}
\end{figure}

\begin{figure}
\includegraphics[width=15cm]{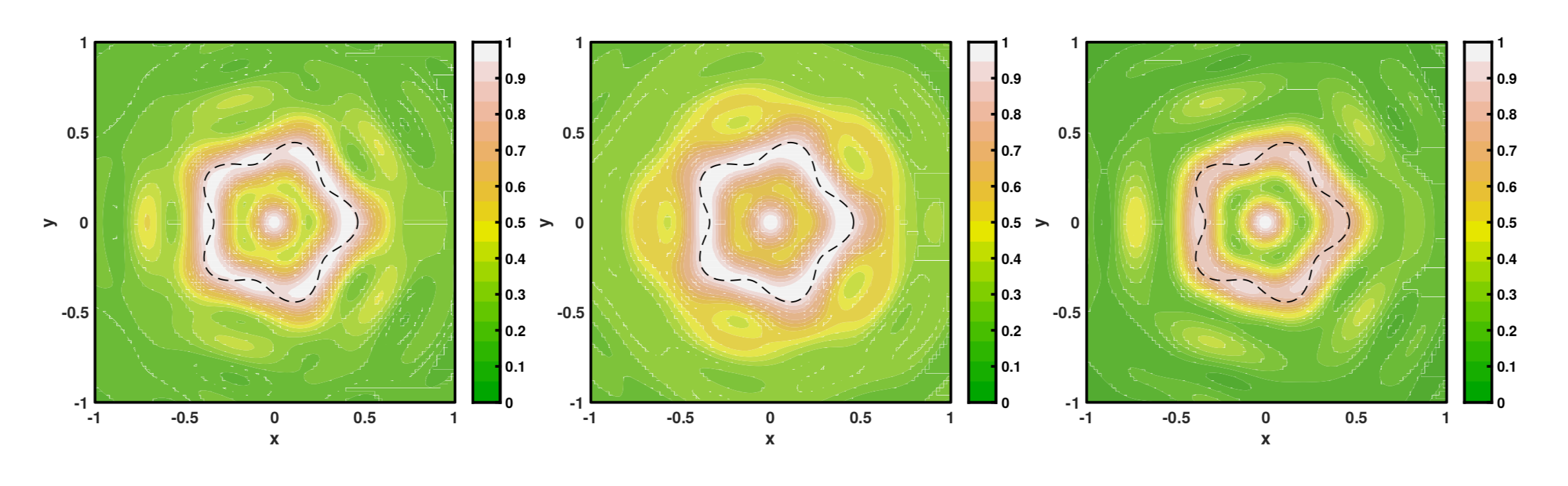}
\caption{{Reconstructions of the star-shaped domain using the three indicator functions  $\widetilde{W}_{\text{DSM}}(z)$, $\widetilde{W}_{\text{FDSM}}(z)$, and $\widetilde{W}_{\text{TDSM}}(z)$ with noise level $\delta=5\%$.}}
\label{recon2}
\end{figure}

\begin{figure}
\includegraphics[width=15cm]{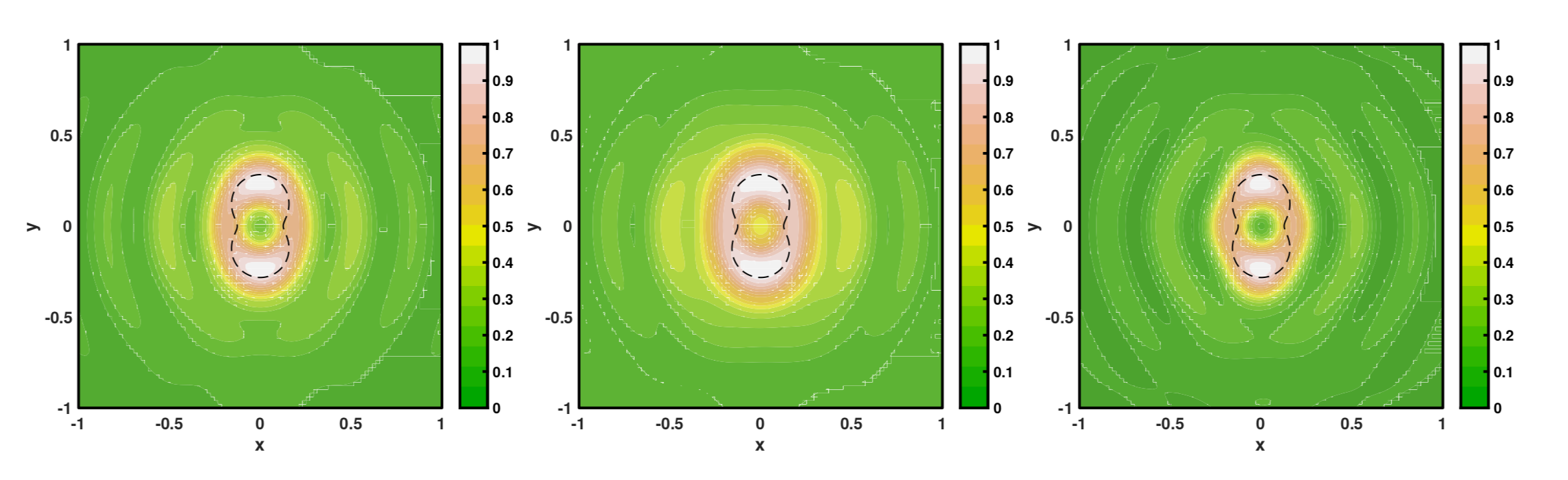}
\caption{Reconstructions of the peanut-shaped domain using the three indicator functions  $\widetilde{W}_{\text{DSM}}(z)$, $\widetilde{W}_{\text{FDMS}}(z)$, and $\widetilde{W}_{\text{TDMS}}(z)$ with noise level $\delta=5\%$.}
\label{recon3}
\end{figure}

We additionally show the robustness of the newly proposed indicator based on the Tikhonov regularization for the peanut-shaped scatterer. We use the noise levels $\delta=1\%$, $5\%$, and $10\%$ for the far-field data and show the results for the indicator approximation $\widetilde{W}_{\text{TDSM}}(z)$ in Figure \ref{noise2d}. As we can see, the reconstructions of the peanut-shaped scatterer are robust with respect to these noise levels. Almost no difference can be seen visually. {\color{black} It is well-known that the robustness is tied to the choice of the regularization parameter which is not studied here. However, in practice one would like to pick an optimal $\alpha$ depending on the noise level $\delta$ which is left for future investigation.}

\begin{figure}
\includegraphics[width=15cm]{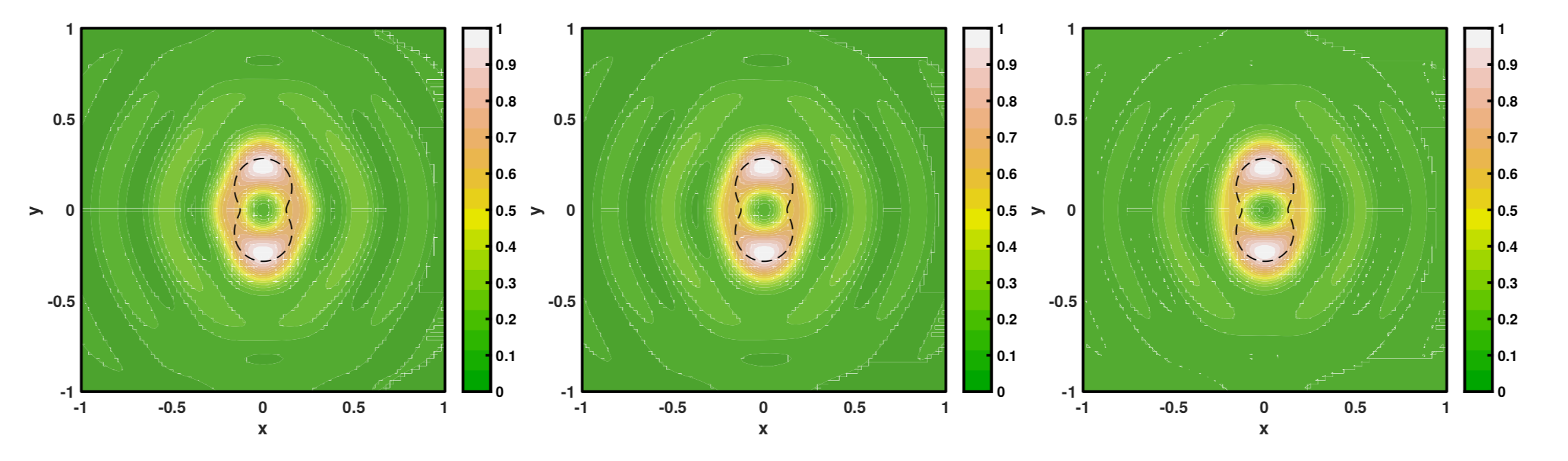}
\caption{Reconstructions of the peanut-shaped domain using the indicator function $\widetilde{W}_{\text{TDSM}}(z)$ with noise level $\delta=1\%$, $5\%$, and $10\%$.}
\label{noise2d}
\end{figure}

\subsection{Numerical results in three dimensions}
We now give some numerical examples in $\R^3$ for the case of $A= \alpha I$ where both $\alpha$ and  $n$ are constants. In order to compute the far-field pattern we use boundary integral 
equations derived from Green's representation formula for the scattered field $u^s$ in $\R^3 \setminus \ov{D}$ and the total field $u$ in $D$ (see \cite{BIE-ref} for details). 
This gives a 2$\times$2 system of boundary integral equations on the boundary $\partial D$ for the densities $\varphi_1$ and $\varphi_2$. 
We employ a boundary element collocation method to solve the system numerically. Here, the far-field operator is approximated for simplicity by constant interpolation over a triangulation of the unit sphere. 
For a possible higher-order approximation of it, we refer the reader to \cite{int-approx}. 

In our numerical examples the computed densities $\varphi_1$ and $\varphi_2$ depend on the incident direction $\hat{y}$ since each incident direction corresponds to a new right-hand side for the  boundary integral equations. The corresponding far-field pattern $u^\infty$ can be derived from Green's representation and is given by 
\[u^\infty(\hat{x}, \hat{y} )= \frac{1}{4 \pi} \int_{\partial D} \varphi_1 (w, \hat{y}) {\partial_{\nu(w)} } \mathrm{e}^{-\mathrm{i}k w \cdot \hat{x} } -  \varphi_2 (w, \hat{y}) \mathrm{e}^{-\mathrm{i}k w \cdot \hat{x}} \, \mathrm{d}s(w)\,.\]
The boundary of the domain is given in spherical coordinates such that 
\[\partial D = r(\phi) \left(\sin(\phi) \cos (\theta), \sin(\phi) \sin(\theta) , \cos(\phi) \right)\]
 where $r(\phi)$ is given by $r(\phi)=1$, 
\begin{eqnarray*}
r(\phi) = \frac{3}{2} \sqrt{\cos^2(\phi) +\frac{1}{4} \sin^2(\phi) } \quad \text{and } \quad r(\phi) &=&\frac{3}{5} \sqrt{\frac{17}{4} +2 \cos(3\phi) } 
\end{eqnarray*}
which represent a sphere, peanut-shaped, and acorn-shaped domain, respectively. See Figure \ref{ref-domain} for a triangularization of their surfaces. 
{Note that some of the far-field data have already been used in \cite[Section 8]{pele}.}
\begin{figure}
\centering
\includegraphics[width=15cm]{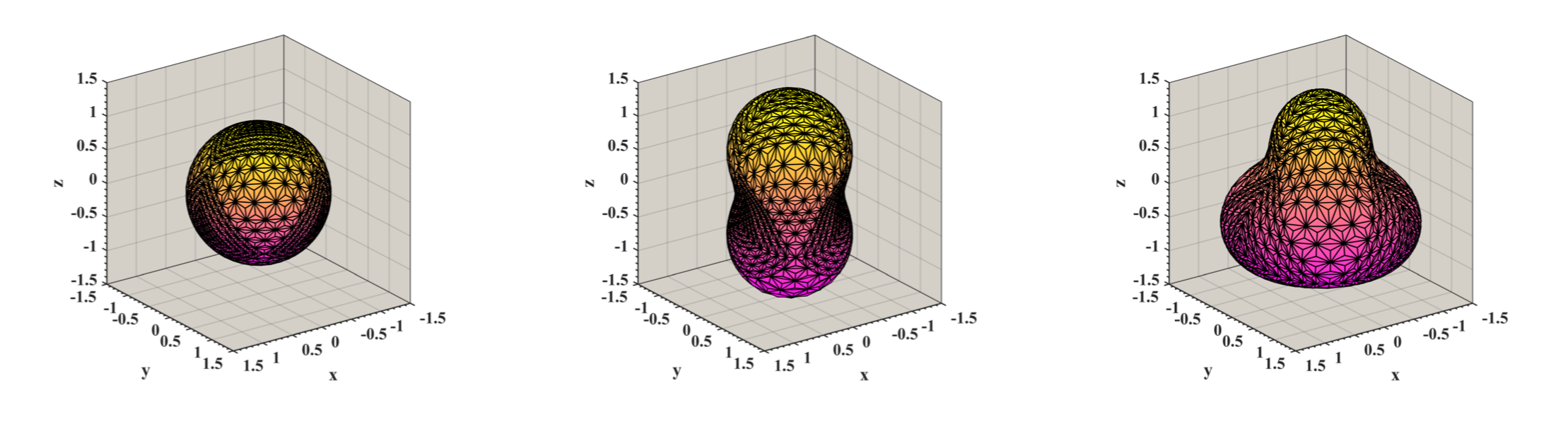}
\caption{Triangularization of the sphere, peanut-shaped, and acorn-shaped domains with 512 faces.}
\label{ref-domain}
\end{figure}

In Figures \ref{recon4}--\ref{recon6} we fix the wave number $k=2$ with $n=1/2$ and $A=2I$ in our examples. We add $5\%$ random noise to the approximated far-field pattern which corresponds to $\delta =0.05$. Here, we use $M=258$ incidence and observation directions that are {`almost' uniformly spaced on the unit sphere (see \cite[Appendix A.1]{akak} for details)}. We plot the three indicators on the $y - z$ plane such that the sampling region is given by $[-2 , 2] \times [-2 , 2]$ where the sampling points are taken to be $100 \times 100$ equally spaced points in the sampling region. Computing the three indicator functions takes roughly five to six seconds for each of the domains under consideration. In Figures  \ref{recon4}--\ref{recon6}  we contour plots of the three indicator functions for spherical, peanut-shaped, and acorn-shaped domain. Again we see that the three indicators give comparable reconstructions for the sphere and peanut-shaped obstacle. In Figure \ref{recon6} we see that the indicator $\widetilde{W}_{\text{TDSM}}(z)$ gives a more detailed reconstruction of the acorn-shaped scatterer. 

\begin{figure}
\includegraphics[width=15cm]{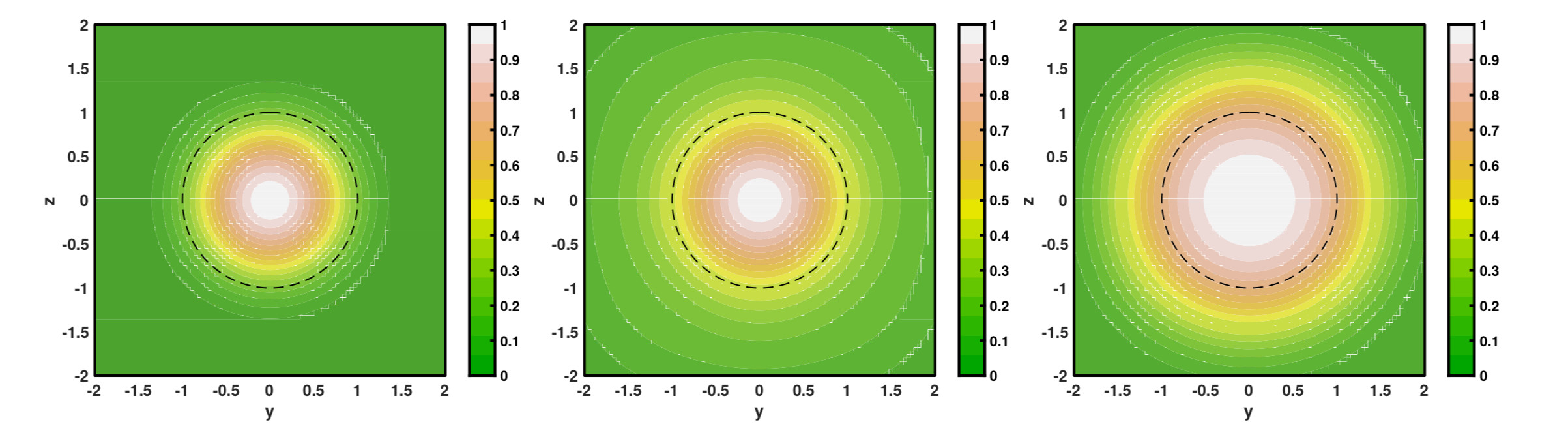}
\caption{{Reconstructions of the spherical domain using the three indicator functions  $\widetilde{W}_{\text{DSM}}(z)$, $\widetilde{W}_{\text{FDSM}}(z)$, and $\widetilde{W}_{\text{TDSM}}(z)$ with noise level $\delta=5\%$.}}
\label{recon4}
\end{figure}

\begin{figure}
\includegraphics[width=15cm]{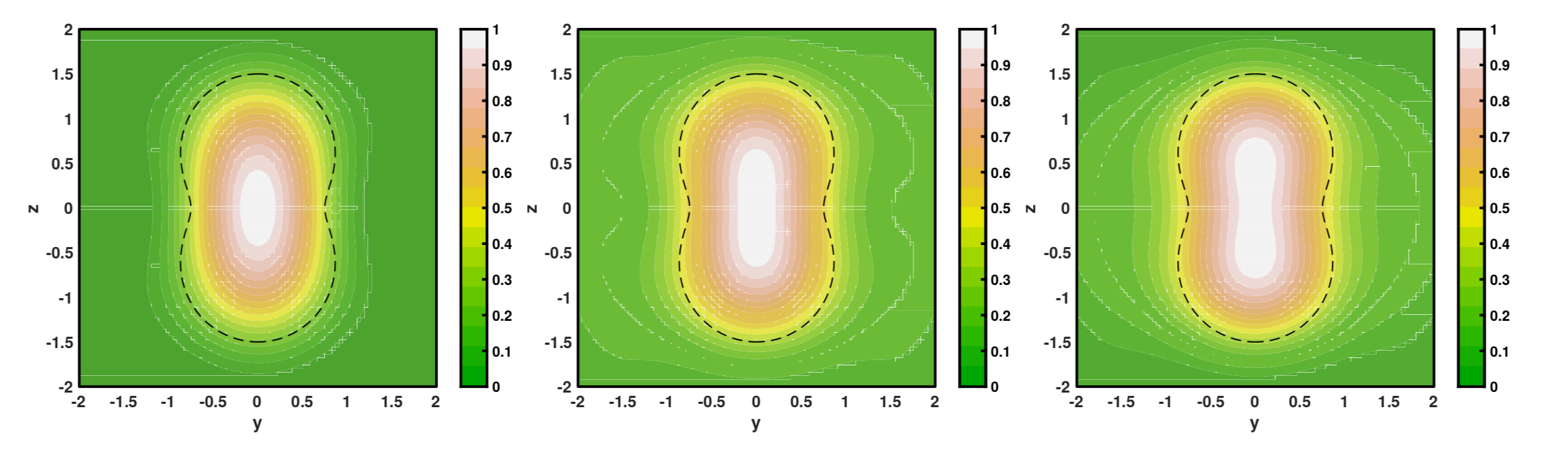}
\caption{{Reconstructions of the peanut-shaped domain using the three indicator functions  $\widetilde{W}_{\text{DSM}}(z)$, $\widetilde{W}_{\text{FDSM}}(z)$, and $\widetilde{W}_{\text{TDSM}}(z)$ with noise level $\delta=5\%$.}}
\label{recon5}
\end{figure}

\begin{figure}
\includegraphics[width=15cm]{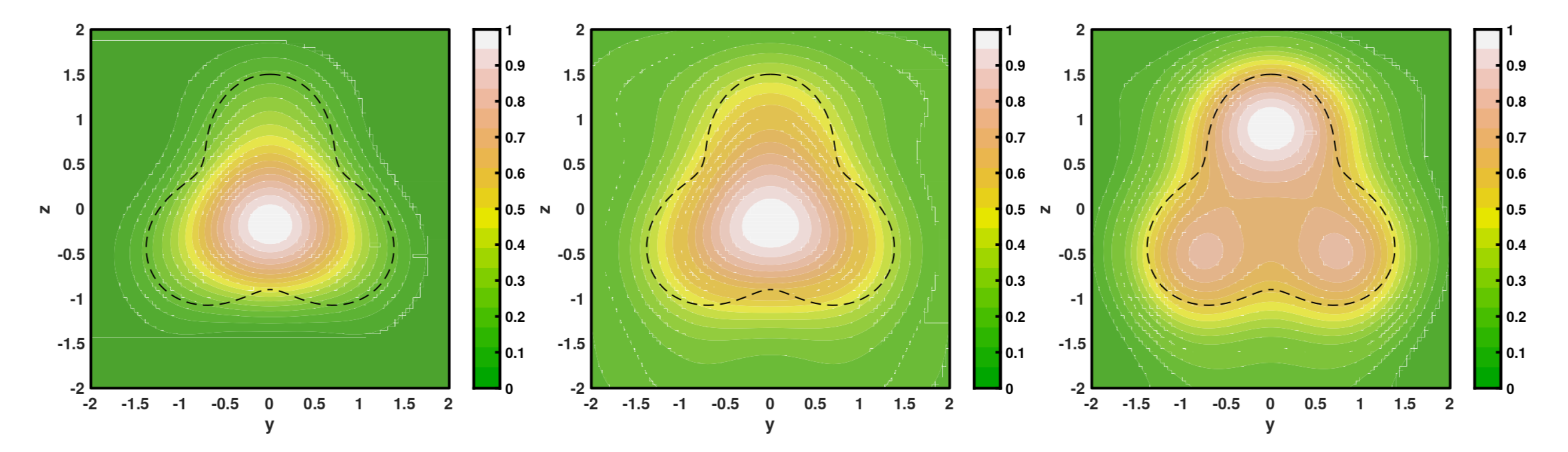}
\caption{{Reconstructions of the acorn-shaped domain using the three indicator functions  $\widetilde{W}_{\text{DSM}}(z)$, $\widetilde{W}_{\text{FDSM}}(z)$, and $\widetilde{W}_{\text{TDSM}}(z)$ with noise level $\delta=5\%$.}}
\label{recon6}
\end{figure}

Finally, we show the reconstructions for the peanut-shaped domain by the newly proposed indicator based on the Tikhonov regularization using the noise levels $\delta=1\%$, $5\%$, and $10\%$ for the far-field data. The results for the indicator approximation $\widetilde{W}_{\text{TDSM}}(z)$ are given in Figure \ref{noise3d}. Here, we can see that we obtain the best reconstruction for the noise level $1\%$. The reconstruction deteriorates significantly as the noise level increases.

\begin{figure}
\includegraphics[width=15cm]{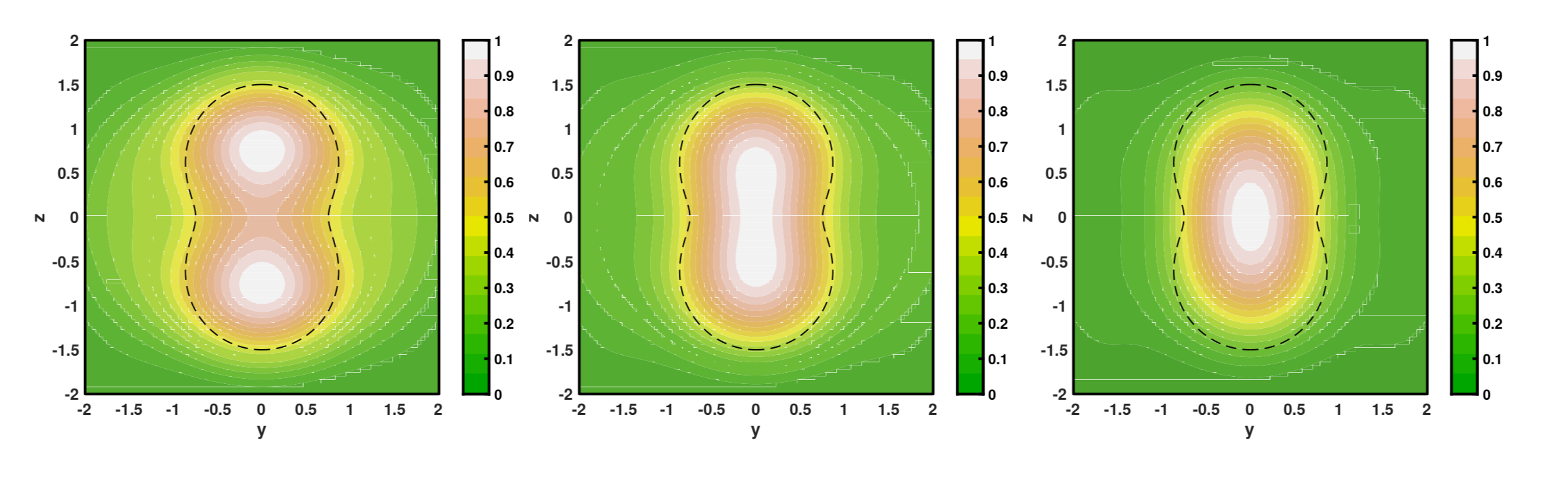}
\caption{Reconstructions of the peanut-shaped domain using the indicator function $\widetilde{W}_{\text{TDSM}}(z)$ with noise level $\delta=1\%$, $5\%$, and $10\%$.}
\label{noise3d}
\end{figure}

\section{Summary and outlook}\label{summary}
In this article, we have studied three DSM indicator functions using the theoretical basis of the FM for anisotropic materials with real-valued coefficients. 
The equivalence of two indicators previously studied is shown as well as a new indicator based on the Tikhonov regularization applied to the far-field equation for the factorization method is proposed.
Precisely, we are able to prove that one of the indicator functions decays as the sampling point moves away from the scatterer. Note that the results here are stated for anisotropic materials but the analysis is valid whenever the scattering operator is unitary and the far-field operator is injective with dense range which means that these results hold for a wide range of scattering objects. Numerical examples show the validity of the three indicator functions. In our experiments we are able to reconstruct the scatterers in seconds which gives that these DSMs are computationally cheap and  rigorously justified. The analysis in this article depends mainly on the orthonormal spectral decomposition of the far-field operator.  It is known that for the case of complex-valued coefficients the far-field operator fails to be normal and therefore does not have the orthonormal spectral decomposition which is vital in our approach. However, one can still apply the FM 
either to the operators $\mathrm{Im}(F) $  or $ F_{\sharp} = |\mathrm{Re}(F) | + |\mathrm{Im}(F) |$ (see for e.g. \cite{fm-gbc,akak,kirschbook,armin}). By definition one has that these operators are self-adjoint and compact which implies they have an orthonormal spectral decomposition. One can then construct DSM indicator functions using either $\mathrm{Im}(F) $  or $ F_{\sharp}$ also for the case of scatterers with complex-valued coefficients.


\end{document}